\newif\ifIncProofs

\newif\ifnotIncProofs
\ifIncProofs\notIncProofsfalse\else\notIncProofstrue\fi

\documentclass[prd,aps,amsfonts,showpacs,longbibliography,notitlepage,twocolumn,groupedaddress,superscriptaddress]{revtex4-2}
\usepackage{tikz} 
\usepackage{amsmath,amsfonts,amsthm, amssymb, amsthm, mathtools}
\usepackage{placeins}
\usepackage{dsfont}
\usepackage{csquotes}
\usepackage{amssymb}
\usepackage{verbatim}
\usepackage{bm}
\usepackage{stmaryrd}
\usepackage[caption=false]{subfig}
\usepackage{physics} 

\usepackage{hyperref}
\definecolor{darkred}  {rgb}{0.5,0,0}
\definecolor{darkblue} {rgb}{0,0,0.5}
\definecolor{darkgreen}{rgb}{0,0.5,0}
\hypersetup{
  colorlinks = true,
  urlcolor  = blue,         
  linkcolor = darkblue,     
  citecolor = darkgreen,    
  filecolor = darkred       
}

\theoremstyle{definition}

\newtheorem{corollary}{Corollary}
\newtheorem{definition}{Definition}

\newtheorem{lemma}{Lemma}
\newtheorem{proposition}{Proposition}
\newtheorem{theorem}{Theorem}

\newtheorem{remark}{Remark}

\newcommand{\GL}{\mathrm{GL}}
\newcommand{\U}{\mathrm{U}}

\newcommand{\mc}{\mathcal}
\newcommand{\mb}{\mathbb}
\newcommand{\supp}{\mathrm{supp}}
\newcommand{\Haar}{\mathrm{Haar}}
\newcommand{\diag}{\mathrm{diag}}
\newcommand{\poly}{\mathrm{poly}}

\begin{document}
\title{A Quantum Method of Types} 
\author{Arick Grootveld} 
\affiliation{Department of Electrical Engineering \& Computer Science, Syracuse University, Syracuse, NY, USA} \affiliation{Institute for Quantum \& Information Sciences, Syracuse University, Syracuse, NY, USA} 
\thanks{We thank B. Chen, V. Gandikota, J. Pollack, A. Maloney, P. Wu, and H. Yang for insightful conversations during the development of this manuscript.}

\begin{abstract} 

The method of types is a fundamental tool in classical information theory, with applications ranging from composite hypothesis testing and universal source coding to the capacity of arbitrarily varying channels. In this work we introduce an empirical operator acting as a quantum analog of the empirical distribution. 
We show that this empirical operator satisfies combinatorial and large-deviation bounds, which in combination describe a quantum method of types. As an application, we use our method to prove a universal achievability result for composite quantum hypothesis testing. 

\end{abstract} 
\maketitle

\section{Introduction}

The empirical distribution is a central object in classical probability and statistics, valued in part because it gives a data-driven description of a sample, requiring no assumptions or side information. 
The method of types in classical information theory comprises a set of bounds on empirical events, using the fact that the number of empirical distributions grows polynomially in the number of samples. 
The method of types is highly general, and is especially useful to address non-parametric problems in information theory \cite{Csiszar1998method}. 

Several works have proposed quantum versions of the method of types \cite{Harrow2005SchurTransform, Notzel2014Invariant, notzel2015class} with polynomially many outcomes. Unfortunately, these techniques either assume a basis of interest or focus solely on recovering the spectrum of the operator, meaning that they either require prior knowledge or do not estimate a complete description of the operator. 

Quantum state tomography is a leading candidate for a quantum analog of the empirical distribution, because it matches some of the important qualities: tomography protocols estimate all properties of the state at once, and require no prior knowledge. 
However, current optimal tomography schemes \cite{Keyl_2006, ODonnell2021, Haah_2017, o2016efficient, ODonnel2017EfficientII, pelecanos2025mixed} allow arbitrary unitaries when estimating the operator's eigenbasis, so that for finite $n$ there are uncountably many `empirical' operators.

We propose a discretization of Keyl's tomography protocol using unitary designs, and large-deviation behaviour governed by Keyl's rate (the reverse relative entropy). These empirical operators preserve the important traits of their classical counterpart, namely the number of empirical operators grows polynomially in $n$ and they become dense in the set of density operators, which we interpret as a quantum method of types. We demonstrate the utility of this method by establishing a quantum version of Sanov's theorem and an achievability result for general composite quantum hypothesis testing problems. Our result was inspired by the large-deviation treatments of tomography \cite{Keyl_2006, GMPW_2026} and the quantum empirical distribution described in \cite{hayashi2025another}.




\vspace{-2em}
\section{Preliminaries}

{\small
For background on the classical method of types, see \cite{Csiszar1998method}.
}

\vspace{-2em}

\subsection{Notation}
\label{subsec:Notation}

We denote by
\[
    \mc D_d
    :=
    \{\rho\in M_d(\mb C): \rho\ge 0,\ \Tr(\rho)=1\}
\]
the set of quantum states on \(\mb C^d\). We also denote by
\(\U(d)\) the group of unitary operators on \(\mb C^d\), by \(\GL(d)\) the group of invertible linear operators on \(\mb C^d\), by $\mc L(\mb C^d)$ the set of all linear operators, by $\mc B(\mb C^d)$ the set of bounded operators and by $\mc B(X, Y)$ the set of bounded linear maps from $X \to Y$.

Let
\[
    \Delta_d
    :=
    \left\{
        x=(x_1,\ldots,x_d)\in\mb R^d:
        x_i\ge 0,\ \sum_{i=1}^d x_i=1
    \right\}
\]
be the probability simplex, and let
\[
    \Delta_d^\downarrow
    :=
    \left\{
        x\in\Delta_d:
        x_1\ge x_2\ge \cdots \ge x_d
    \right\}
\]
be the set of probability vectors arranged in decreasing order. For any
vector \(\alpha\in\mb R^d\), we write \(\alpha^\downarrow\) for the vector
obtained by rearranging the entries of \(\alpha\) in decreasing order.
    
For any matrix $A$, define $\Delta_k(A)$ to be the determinant of the leading \(k\times k\) principal minor of $A$: suppose $A = \sum_{i,j=1}^d a_{ij} |i\rangle \langle j|$, denote $A_{k\times k} = \sum_{i,j=1}^k a_{ij} |i\rangle \langle j|$ and 
\begin{equation}\label{eq:determinant-principle-minor}
    \Delta_k(A) = \det A_{k\times k}.
\end{equation}

We will make use of the following divergences. Let $F$ be the quantum fidelity, $\|\cdot\|_1 $ the trace norm, and $D(\cdot\|\cdot)$ be the quantum relative entropy, so that 
\begin{align*}
    F(\sigma, \rho) &= \left(\Tr\left[ \left(\sigma^{1/2} \rho \sigma^{1/2}\right)^{1/2} \right]\right)^2\\
    \| \sigma - \rho \|_1 &= \Tr\left[\sqrt{(\sigma - \rho)^\dagger (\sigma - \rho)}\right] \\
    D(\sigma\|\rho) &= \begin{cases}
        \Tr[\sigma (\log \sigma - \log \rho)], & \supp(\sigma) \subseteq \supp(\rho)\\
        \infty, & \text{otherwise}
    \end{cases}
\end{align*}

We use $D_R(\sigma\|\rho)$ to denote the reverse quantum relative entropy \cite{Keyl_2006,audenaert2015alpha, hayashi2026operational, GMPW_2026}, so that for $\sigma = U \diag(x) U^\dagger$ where $x \in \Delta_d^\downarrow$, $U \in U(d)$, when $\rank(\sigma) = \rank(\sigma \rho)$,
{\small 
\begin{align}
    \label{eq:reverseRelativeEntropy_Def}
    D_R(\sigma\|\rho) = -H(x) - \sum_{k=1}^d (x_k - x_{k+1}) \log\Delta_k\left(U^\dagger \rho U\right),\\
    \text{ where } H(x) = -\sum_{k=1}^d x_k \log x_k, \qquad x_{d+1} := 0
\end{align}
}
and otherwise $D_R(\sigma\|\rho) = \infty$. 

We gather several properties of $D_R$ from the literature. 
\begin{remark}[Properties of the reverse relative entropy]
    \label{remark:PropsOfDR}
    \leavevmode
    \begin{enumerate}
        \item \label{item:DR_LSC} $D_R(\sigma\|\rho)$ is lower semicontinuous in $\sigma$ and  $\rho$ \cite{Keyl_2006}.

        \item \label{item:DR_ContWhenFinite} $D_R(\cdot \|\rho)$ is continuous on $\{\sigma \in \mc D_d: D_R(\sigma\|\rho) < \infty\}$ \cite[Theorem 3]{audenaert2015alpha}.
    
        \item \label{item:DR_FidRelation}
        From an intermediate result in \cite[Corollary 26]{hayashi2026operational}, we have 
        \begin{equation}
            \label{eq:DR_FidRelation}
            D_R(\sigma\|\rho) \geq -\log F(\sigma, \rho).
        \end{equation}
    \end{enumerate}
\end{remark}

\subsection{Schur-Weyl duality}
\label{subsec:SWDuality}

In this section we provide some useful results from representation theory, which can be found in standard references \cite{FultonHarris1991}. For resources tailored towards quantum information theory, see \cite{Harrow2005SchurTransform,hayashi2017}. 

Let \(d,n\in\mathbb N\), and write \([d]=\{1,\ldots,d\}\). 
We work with the \(n\)-partite Hilbert space
\[
    \mc H_n=(\mathbb C^d)^{\otimes n}.
\]
We use $\lambda \vdash_d n$ to denote a Young diagram with $n$ boxes and at most $d$ rows, and \(\ell(\lambda)\) the number of nonzero elements of \(\lambda\). Define $\Lambda_{n,d} = \{\lambda \vdash_d n\}$, and for each \(\lambda \in \Lambda_{n,d}\), let \((\mathcal P_\lambda,p_\lambda)\) denote the irreducible \(S_n\)-representation indexed by \(\lambda\), where \(\mathcal P_\lambda\) is the Specht module. Let \((\mathcal Q_\lambda,q_\lambda)\) denote the irreducible polynomial \(U(d)\)-representation of highest weight \(\lambda\), where \(\mathcal Q_\lambda\) is the Schur module. Schur--Weyl duality decomposes the tensor space as
\[
    \mathcal H_n
    =
    \bigoplus_{\lambda \vdash_d n}
    \mathcal P_\lambda\otimes \mathcal Q_\lambda .
\]
After choosing orthonormal bases of $\mc P_\lambda$ and $\mc Q_\lambda$, this
decomposition is implemented by a unitary change of coordinates, the
\textit{Schur transform},
\begin{equation}\label{eq:schur-transform-space}
    U_{\mathrm{Schur}}:
    (\mb C^d)^{\otimes n}
    \longrightarrow
    \bigoplus_{\lambda \vdash_d n}
    \mathcal P_\lambda\otimes \mathcal Q_\lambda.
\end{equation}

If one represents the direct sum using an explicit block-label register,  the Schur transform may equivalently be written as, 
\[
    U_{\mathrm{Schur}}
    \bigl(X^{\otimes n}\bigr)
    U_{\mathrm{Schur}}^\dagger
    =
    \sum_{\lambda \vdash_d n}
    |\lambda\rangle\langle\lambda|
    \otimes \textbf{1}_{\mc P_\lambda}\otimes q_\lambda(X),
\]
for \(X \in \mc L(\mathbb C^d)\).

It can be helpful to think of a Young diagram, $\lambda$, as an ordered vector of length $d$, whose elements sum to $n$. 
Indeed, the result of normalizing $\lambda$ via $\bar \lambda = \frac{1}{n} \lambda$, corresponds to an ordered empirical distribution with denominator $n$ \cite{Harrow2005SchurTransform}. Let $T_{n,d}$ denote the set of (unordered) empirical distributions, i.e. the set of classical types. With the standard action of the permutation group we have $S_d \Lambda_{n,d} \cong T_{n,d}$. 
We remark that as $n \to \infty$, $T_{n,d}$ becomes dense in $\Delta_d$ \cite{Cover2006elements}.
\begin{lemma}[Empirical distributions become dense]
    \label{lem:TypesAreDense}
    \[
        \lim_{n \to \infty} \sup_{p \in \Delta_d} \inf_{ q \in T_{n,d}} \|p -  q\|_1 = 0.
    \]
\end{lemma}

We will make use of the following lemmas.
\begin{lemma}[Schur orthogonality]\label{lemma:Schur-orthogonality}
Fix $\lambda \vdash_d n$. For every \(X_\lambda\in \mc L(\mathcal Q_\lambda)\),
\begin{equation}
    \int_{\mathcal U(d)}
    q_\lambda(U)\,X_\lambda\,q_\lambda(U)^\dagger\, d\mu_{\Haar}(U)
    =
    \frac{\Tr(X_\lambda)}{\dim \mathcal Q_\lambda}\,
    \mathbf 1_{\mathcal Q_\lambda}.
\end{equation}
\end{lemma}

We note several standard estimates for the size of the dimensions of irreducible representations.
\begin{lemma}[Representation-theoretic estimates]
\label{lemma:rep-data}

For \(\lambda\vdash_d n\), let \(\bar\lambda:=\lambda/n\). Then
\begin{equation}\label{eq:dimension-estimates}
    1 \leq \dim \mathcal Q_\lambda
    \le
    (n+1)^{d(d-1)/2},
\end{equation}
as well as
\begin{equation}\label{eq:specht-dimension-estimates}
    (n+d)^{-d(d+1)/2} e^{nH\left(\bar\lambda\right)}
    \le
    \dim \mathcal P_\lambda
    \le
    e^{nH\left(\bar\lambda\right)}.
\end{equation}
\end{lemma}

We use \(|\phi_\lambda\rangle\) to denote the highest-weight vector of the
irreducible representation \(q_\lambda\). The standard highest-weight matrix coefficient formula \cite[Sect. IX.8]{simon1996representations} gives us 
\begin{lemma}[Highest Weight Projection]
    \label{lem:HighestWeightProj}
    For \(\tau \in \mc D_d\), 
    \begin{equation}\label{eq:highest-weight-minor-formula}
        \bra{\phi_\lambda} q_\lambda(\tau) \ket{\phi_\lambda}
        =
        \prod_{k=1}^d
        \Delta_k(\tau)^{\lambda_k-\lambda_{k+1}},
        \qquad
        \lambda_{d+1}:=0.
    \end{equation}
\end{lemma}

\subsection{Exact unitary designs}
\label{subsection:Prelims_NDesigns}

A useful property of classical empirical distributions is that the number of possible outcomes grows only polynomially fast in $n$. Unfortunately, standard approaches of estimating the eigenspace of an operator use Haar random unitaries for eigenvector estimation, leading to POVMs with uncountably many outcomes. This can be remedied by replacing Haar-continuous eigenspace estimation with a finite, exactly normalized measurement. The appropriate notion is that of an exact unitary design \footnote{Our definition of an exact unitary $t$-design differs from standard definitions, e.g. \cite{Roy2009Unitary}. Specifically, we require that the design can exactly integrate arbitrary polynomials with bounded degree in $U, U^\dagger$, while the standard definition only requires this to be fulfilled for balanced homogeneous polynomials (i.e. polynomials with matching degree in $U$ and $U^\dagger$).}.
\begin{definition}[Exact unitary design]
Let \(t\in\mathbb N\). A finite set
\[
    \mc V_t\subset \U(d)
\]
is called an exact unitary \(t\)-design if, for any polynomial,
\(f:\U(d)\to\mathbb C\), with degree at most \(t\) in the entries of \(U\) and degree at most \(t\) in \(U^\dagger\), 
\begin{equation}\label{equation:NDesign_PolyChar}
    \frac{1}{|\mc V_t|}
    \sum_{V\in\mc V_t} f(V)
    =
    \int_{\U(d)} f(U)\,d\mu_{\Haar}(U).
\end{equation}
Equivalently, for every \(0\le m,r\le t\) and every
\(A\in\mc B((\mathbb C^d)^{\otimes r}, (\mb C^d)^{\otimes m})\),
\begin{equation}\label{equation:NDesign_HermChar}
    \frac{1}{|\mc V_t|}
    \sum_{V\in\mc V_t}
    V^{\otimes m}A(V^\dagger)^{\otimes r}
    =
    \int_{\U(d)}
    U^{\otimes m}A(U^\dagger)^{\otimes r}
    \,d\mu_{\Haar}(U).
\end{equation}
\end{definition}

Combining \eqref{equation:NDesign_HermChar} with Lemma~\ref{lemma:Schur-orthogonality}, one immediately has: 
\medskip
\begin{lemma}[Design twirling on irreducible Schur blocks]
\label{lemma:design-twirling-Schur-block}
Let \(\mc V_n\subset\U(d)\) be an exact unitary \(n\)-design. Then, for
every \(\lambda\vdash_d n\) and every \(X\in\mc B(\mc Q_\lambda)\),
\begin{equation}\label{eq:design-twirl-qlambda}
    \frac{1}{|\mc V_n|}
    \sum_{V\in\mc V_n}
    q_\lambda(V)Xq_\lambda(V^\dagger)
    =
    \frac{\Tr X}{\dim\mc Q_\lambda}\,
    \mathbf 1_{\mc Q_\lambda}.
\end{equation}
\end{lemma}

We need the following two facts about the asymptotic nature of exact designs. First, any sequence of exact designs of increasing order becomes dense in
\(\U(d)\).
\medskip
\begin{lemma}[Exact designs become dense]
\label{lemma:DesignsAreDense}
Let \(\{\mc V_n\}_{n\ge1}\) be a sequence such that \(\mc V_n\) is an
exact unitary \(n\)-design. Then
\[
   \lim_{n\to \infty} \sup_{U\in\U(d)}
    \inf_{V\in\mc V_n}
    \|U-V\|_1 = 0.
\]
\end{lemma}
\begin{proof}
    The proof follows from \eqref{equation:NDesign_PolyChar}, and an elementary application of the Stone-Weierstrass theorem. 
\end{proof}


Finally, the designs may be chosen with polynomial cardinality, see \cite{etayo2018asymptotically}:
\medskip

\begin{lemma}[Polynomial exact unitary designs]
    \label{lemma:polySizeUnitaryDesigns}
    For fixed \(d\), there exist exact unitary \(n\)-designs
    \(\mc V_n\subset\U(d)\) and constants \(C_d,c_d>0\) such that
    \begin{equation}\label{eq:design-polynomial-size}
        |\mc V_n|\le C_d n^{c_d}.
    \end{equation}
    In particular,
    \begin{equation}\label{eq:design-subexp-size}
        \lim_{n\to\infty}
        \frac1n\log |\mc V_n|
        =
        0.
    \end{equation}
\end{lemma}
Notably, applying results in \cite{kane2015small}, one can recover the explicit (though suboptimal) constants 
\begin{equation}
    \label{eq:UDesigns_ExplicitConstants}
    C_d = 9d^8, \qquad c_d=4d^2
\end{equation}

\section{Quantum Types}

The goal of this section is to prove the following theorem.
\begin{theorem}[Quantum Types]
\label{theorem:QuantMethodOfTypes}
    There exists a sequence $\{(\hat{\Sigma}_n, \mc M_n)\}_{n \geq 1}$ where for each $n$, $\hat \Sigma_n \subset \mc D_d$ and $\mc M_n = \{M_\sigma^{(n)}\}_{\sigma \in \hat \Sigma_n} \subset \mc L_+(\mc H_n)$ is a POVM, such that:
    \begin{enumerate}
        \item (Density): As $n \to \infty$, $\hat \Sigma_n$ becomes dense in $\mc D_d$:
        \begin{equation}
            \label{eq:empOpsDense}
            \lim_{n \to \infty} \sup_{\tau \in \mc D_d} \inf_{\sigma \in \hat \Sigma_n} \|\tau - \sigma\|_1 = 0.
        \end{equation}
        \item (Exponential Rate): For any $\rho \in \mc D_d$, and $\sigma \in \hat \Sigma_n$ such that $D_R(\sigma\|\rho) < \infty$,
        {
        \begin{equation}
            \label{eq:QTypeProbBound}
            \begin{split}
                \left[9d^8 (n+d)^{\frac{9d^2 + d}{2}}\right]^{-1} e^{-n D_R(\sigma\|\rho)} 
                \leq \Tr[M_\sigma^{(n)} \rho^{\otimes n}] \\
                \leq (n+1)^{d^2} e^{-n D_R(\sigma\|\rho)}
            \end{split}
        \end{equation}
        }
        Additionally, if $D_R(\sigma\|\rho) = \infty$, then 
        \begin{equation}
            \label{eq:QTypeProbZero}
            \Tr[M_\sigma^{(n)} \rho^{\otimes n}] = 0.
        \end{equation}
        \item (Polynomial size): $\forall n \in \mathbb{N}$, 
        \begin{equation}
            \label{eq:QTypes_polySize}
            |\hat \Sigma_n| \leq 9d^8(n+1)^{5d^2}
        \end{equation}
    \end{enumerate}
\end{theorem}


To establish Theorem \ref{theorem:QuantMethodOfTypes}, we build a POVM with outcomes that correspond to possible empirical operators. Let $\mc V_n$ be an exact unitary $n$-design with $|\mc V_n| \leq 9d^8 n^{4d^2}$, and define $\hat \Gamma_n = \{(\lambda, U): \lambda \in \Lambda_{n,d}, U \in \mc V_n\}$. Our POVM will be $\{M^{(n)}_{\lambda, U}\}_{(\lambda, U) \in \hat \Gamma_n}$, with 
{\footnotesize
\begin{align*}
    M^{(n)}_{\lambda, U} 
    = \frac{\dim \mc Q_\lambda}{|\mc V_n|} U_{\rm{Schur}}^\dagger \left(\ket{\lambda}\bra{\lambda} \otimes \textbf{1}_{\mc P_\lambda} \otimes q_\lambda(U)  \ket{\phi_\lambda}\bra{\phi_\lambda} q_\lambda(U^\dagger) \right)U_{\rm{Schur}}
\end{align*}
}
It is not hard to verify that $\{M^{(n)}_{\lambda, U}\}_{(\lambda, U) \in \hat \Gamma_n}$ is a valid POVM. 
{\small
\begin{align*}
    &U_{\rm{Schur}} \left(\sum_{(\lambda, U) \in \hat \Gamma_n} M^{(n)}_{\lambda, U} \right) U_{\rm{Schur}}^\dagger \\
    &= \sum_{\lambda \vdash_d n} \frac{\dim \mc Q_\lambda}{|\mc V_n|} \sum_{U \in \mc V_n} \ket{\lambda}\bra{\lambda} \otimes \textbf{1}_{\mc P_\lambda} \otimes q_\lambda(U)  \ket{\phi_\lambda}\bra{\phi_\lambda} q_\lambda(U^\dagger)\\
    &= \sum_{\lambda \vdash_d n} \ket{\lambda}\bra{\lambda} \otimes \textbf{1}_{\mc P_\lambda} \otimes \left[\frac{\dim \mc Q_\lambda}{|\mc V_n|}\sum_{U \in \mc V_n} q_\lambda(U)  \ket{\phi_\lambda}\bra{\phi_\lambda} q_\lambda(U^\dagger) \right]\\
    &\overset{(a)}{=} \sum_{\lambda \vdash_d n} \ket{\lambda}\bra{\lambda} \otimes \textbf{1}_{\mc P_\lambda} \otimes \textbf{1}_{\mc Q_\lambda}\\
    &= \textbf{1}_{\mc H_n},
\end{align*}
}
where $(a)$ uses Lemma \ref{lemma:design-twirling-Schur-block}.

Now, we can evaluate the probability of an outcome when the POVM is applied to $\rho^{\otimes n}$. For $(\lambda, U) \in \hat \Gamma_n$,
\begin{align*}
    &\Tr[M_{\lambda, U} \rho^{\otimes n}] \\
    &= \Tr\left[U_{\rm{Schur}} M_{\lambda, U}^{(n)} \rho^{\otimes n} U_{\rm{Schur}}^\dagger\right]\\
    &= \frac{\dim \mc Q_\lambda}{|\mc V_n|} \Tr\left[\ket{\lambda}\bra{\lambda} \otimes \textbf{1}_{\mc P_\lambda} \otimes q_\lambda(U) \ket{\phi_\lambda}\bra{\phi_\lambda} q_\lambda(U^\dagger) q_\lambda(\rho)\right]\\
    &= \frac{\dim \mc P_\lambda \dim \mc Q_\lambda}{|\mc V_n|} \prod_{k=1}^d \Delta_k(U^\dagger \rho U)^{\lambda_k - \lambda_{k+1}}.
\end{align*}
The last equality follows from Lemma \ref{lem:HighestWeightProj}. 

We can bound this probability on both sides using Lemma \ref{lemma:rep-data} and Lemma \ref{lemma:polySizeUnitaryDesigns}, along with \eqref{eq:reverseRelativeEntropy_Def},
\begin{align}
    \label{eq:BoundOnMeasuredProb}
    \frac{1}{9d^8 (n+d)^{\frac{9d^2+d}{2}} } e^{-n D_R(U \diag(\bar \lambda) U^\dagger\|\rho)} \leq \Tr[M_{\lambda, U}^{(n)} \rho^{\otimes n}]\\ \leq \frac{1}{|\mc V_n|}(n+1)^{d^2}  e^{-n D_R(U \diag(\bar \lambda) U^\dagger\|\rho)}
\end{align}

This is already close to our desired result; for a pair $(\lambda, U) \in \hat \Gamma_n$, taking $\hat \sigma = U \diag (\bar \lambda) U^\dagger$ yields a valid density operator. 

However, $\sigma$ is not always uniquely defined by a $(\lambda, U)$ pair, because if $\bar \lambda$ has degeneracies, multiple unitaries may induce the same empirical density operator. To remedy this, let \(\hat \Sigma_n = \{\sigma = U \diag(\bar \lambda) U^\dagger \in \mc D_d: (\lambda, U) \in \hat \Gamma_n\}\), and for any $\sigma \in \hat \Sigma_n$ take \(R(\sigma) = \{(\lambda, U) \in \hat \Gamma_n: U \diag(\bar \lambda) U^\dagger = \sigma\}\). 
Then for any $\sigma \in \hat \Sigma_n$, we define 
\begin{equation}
    M_\sigma^{(n)} := \sum_{(\lambda, U) \in R(\sigma)} M_{\lambda, U}^{(n)}.
\end{equation}
Thus, the POVM for our quantum method of types will be $\mc M_n := \{M_\sigma^{(n)}\}_{\sigma \in \hat \Sigma_n}$, satisfying the claims in Theorem \ref{theorem:QuantMethodOfTypes}.
\begin{proof}
    Let $\hat \Sigma_n$, and $\mc M_n$ be defined as above. 
    \begin{enumerate}

        \item \eqref{eq:empOpsDense} is a corollary of Lemma \ref{lemma:DesignsAreDense} and Lemma \ref{lem:TypesAreDense}.
        \ifIncProofs
        More explicitly, for any $\tau = V D V^\dagger \in \mc D_d$ taking the minimizer $\sigma = U \bar \lambda U^\dagger \in \hat \Sigma_n$,  
        \begin{equation*}
            \|\sigma - \tau\|_1 \leq \|U \diag(\bar \lambda) U^\dagger - U D U^\dagger\|_1 + \|U D U^\dagger - V D V^\dagger\|_1 \to 0
        \end{equation*}
        \fi

        \item The lower bound in \eqref{eq:QTypeProbBound} follows from \eqref{eq:BoundOnMeasuredProb}, while the upper bound uses $\frac{|R(\sigma)|}{|\mc V_n|} \leq 1$. On the other hand, \eqref{eq:QTypeProbZero} can be seen directly from Lemma \ref{lem:HighestWeightProj}.

        \item \eqref{eq:QTypes_polySize} uses the fact that $|\mc M_n| \leq |T_{n,d}| |\mc V_n|$, along with \eqref{eq:UDesigns_ExplicitConstants}, and $|T_{n,d}| \leq (n+1)^d$. 
    \end{enumerate}
\end{proof}


\section{Composite Quantum Hypothesis Testing}

We define a state estimation scheme, so that we can use this quantum method of types for composite quantum hypothesis testing. 
\begin{definition}
    \label{def:state-estimation}
    Given a quantum method of types described by $\hat \Sigma_n$ and $\mc M_n$, $E_n$ is the corresponding state estimation scheme (tomography protocol), so that for any $A \subseteq \mc D_d$, we define $E_n(A) := \sum_{\sigma \in  (\hat \Sigma_n \cap A)} M_{\sigma}^{(n)}$, with the understanding that for $A \cap \hat \Sigma_n = \emptyset$, $E_n(A) \equiv 0$.
    \smallskip

    Additionally, for any $\rho \in \mc D_d$ and $A \subseteq \mc D_d$, take
    \begin{equation}
        \label{eq:stateEstProb}
        \mathbb{P}_{\mc M_n}[A| \rho] := \Tr[E_n(A) \rho^{\otimes n}]
    \end{equation}
\end{definition}

We can now compute a large-deviation bound for our empirical density operators.  
\begin{corollary}[Large-deviation bound]
    \label{cor:largeDevBound}
    Take $\hat \Sigma_n$ and $\mc M_n$ as in Theorem \ref{theorem:QuantMethodOfTypes}, and $E_n$ as in Definition \ref{def:state-estimation}, then for any $A \subseteq \mc D_d$, we have
    { \small 
    \begin{align*}
        \left[(n+d)^{\frac{9d^2+d}{2}}9d^8\right]^{-1} \exp\left\{-n \inf_{\sigma \in \hat \Sigma_n \cap A} D_R(\sigma\|\rho)\right\} \leq \mathbb{P}_{\mc M_n}[A| \rho] \\
        \leq (n+1)^{6d^2}  9d^8 \exp\left\{-n \inf_{\sigma \in \hat \Sigma_n \cap A}D_R(\sigma\|\rho)\right\},
    \end{align*}
    }
    with the  convention that $\inf \{\emptyset\} = \infty$.
\end{corollary}

From this we can prove the following version of Sanov's theorem. 
\begin{theorem}[Yet Another Quantum Sanov Theorem]
    \label{thm:yetAnotherQSanov}
    Take $\{\hat \Sigma_n, \mc M_n\}$ as in Theorem \ref{theorem:QuantMethodOfTypes}, and $E_n$ as in Definition \ref{def:state-estimation}. Then for any $A \subset \mc D_d$,
    \begin{align}
        \label{eq:qSanovResult}
        -\inf_{\sigma \in \rm{int}(A)} D_R(\sigma\|\rho) \leq \liminf_{n \to \infty} \frac{1}{n} \log \mathbb{P}_{\mc M_n}[A|\rho] \\
        \leq \limsup_{n \to \infty} \frac{1}{n} \log \mathbb{P}_{\mc M_n}[A|\rho] \leq -\inf_{\sigma \in A}  D_R(\sigma\|\rho).
    \end{align}
\end{theorem}
Here $\rm{int}(A)$ denotes the interior of $A$, and the right hand side does not require the usual closure of $A$, because there are only a finite number of empirical operators for any $n$. 
\begin{proof}
    \ifnotIncProofs
    A proof of this statement follows the same technique used to prove  classical Sanov's theorem for discrete probability distributions, see for example \cite{Dembo2009LargeDeviations}. Both the upper and lower bound come almost immediately from Corollary \ref{cor:largeDevBound}, with the lower bound additionally requiring the continuity property of $D_R$ in Remark \ref{remark:PropsOfDR}.

    \ifIncProofs
    Our proof exactly follows the technique in proving \cite[Theorem 2.1.10]{Dembo2009LargeDeviations}.

    From Corollary \ref{cor:largeDevBound}, we have 
    \[
        -\inf_{\sigma \in \hat \Sigma_n \cap A} D_R(\sigma\|\rho) + o(1) \leq \frac{1}{n} \log \mathbb{P}_{\mc M_n}[A|\rho] \leq -\inf_{\sigma \in \hat \Sigma_n \cap A} D_R(\sigma\|\rho) + o(1),
    \]
    so that 
    \begin{align}
        \limsup_{n \to \infty} \frac{1}{n} \log \mathbb{P}_{\mc M_n}[A|\rho] &= -\liminf_{n \to \infty}\{\inf_{\sigma \in \hat \Sigma_n \cap A} D_R(\sigma\|\rho)\}      \label{align:limsupEqualsNegLiminf}\\
        \liminf_{n \to \infty} \frac{1}{n} \log \mathbb{P}_{\mc M_n}[A|\rho] &= -\limsup_{n \to \infty}\{\inf_{\sigma \in \hat \Sigma_n \cap A} D_R(\sigma\|\rho)\}      \label{align:liminfEqualsNegLimsup}
    \end{align}

    From \eqref{align:limsupEqualsNegLiminf}, and 
    \[
        -\liminf_{n \to \infty} \{\inf_{\sigma \in \hat \Sigma_n \cap A} D_R(\sigma\|\rho)\} \leq \inf_{\sigma \in A} D_R(\sigma\|\rho),
    \]
    we get the upper bound. 

    Now we prove the lower bound. Without loss of generality, assume $\exists \sigma \in \rm{int}(A)$ with $D_R(\sigma\|\rho) < \infty$, as the result holds trivially otherwise. Since $\hat \Sigma_n$ becomes dense in $\mc D_d$, $\exists$ a sequence $\{\sigma_n\}$ so that 
    \[
        \sigma_n \in \hat \Sigma_n, \quad \forall n,
    \]
    $\sigma_n \to \sigma$, and eventually $\sigma_n \in \hat \Sigma_n \cap A$ along with $\rank(\sigma_n) = \rank(\sigma_n \rho)$. 

    Then 
    \begin{align*}
        -\limsup_{n \to \infty} \{\inf_{\tau \in \hat \Sigma_n \cap A} D_R(\tau\|\rho)\} &\geq -\limsup_{n \to \infty} D_R(\sigma\|\rho)\\
        &= D_R(\sigma\|\rho),
    \end{align*}
    where the last equality uses the continuity property of $D_R$, from Remark \ref{remark:PropsOfDR}. 
    
    Therefore, 
    \[
        -\liminf_{n \to \infty} \frac{1}{n} \log \mathbb{P}_{\mc M_n}[A|\rho] = -\limsup_{n \to \infty} \{\inf_{\sigma \in \hat \Sigma_n \cap A} D_R(\sigma\|\rho) \} \geq -\inf_{\sigma \in \rm{int}(A)} D_R(\sigma\|\rho)
    \]
    \fi
\end{proof}

The large-deviation result of Corollary \ref{cor:largeDevBound} immediately implies that $D_R$ is an achievable error exponent for general composite quantum hypothesis testing problems. For $A,B \subset \mc D_d$, the hypothesis testing problem attempts to distinguish between 
\begin{equation}
    \label{eq:HypTestSetup}
    \begin{split}
        H_0&: \omega \in A\\
        H_1&: \omega \in B,
    \end{split}
\end{equation}
given $\omega^{\otimes n}$. A test for this problem is specified by an effect, $0 \leq L_n \leq \textbf{1}$. There are two types of errors for this test, referred to as the type I and type II error  
\begin{align*}
    \alpha_n(L_n) &= \sup_{\sigma \in A} \Tr[(1-L_n) \sigma^{\otimes n}]\\
    \beta_n(L_n) &= \sup_{\rho \in B} \Tr[L_n \rho^{\otimes n}].
\end{align*}

\begin{theorem}[Composite quantum hypothesis testing achievability]
    \label{thm:CompositeQuantumHypTestBound}
    Let $A, B \subset \mc D_d$ be closed sets.

    Using the quantum method of types, we can construct an effect $L_n$ for the hypothesis testing problem
    \[
        H_0: \omega \in A, \qquad
        H_1: \omega \in B
    \]
    such that 
    \begin{equation}
        \alpha_n(L_n) \to 0,
    \end{equation}
    and 
    \begin{align}
        \liminf_{n \to \infty}-\frac{1}{n}\log \beta_n(L_n) \geq \inf_{\substack{\sigma \in A \\ \rho \in B} } D_R(\sigma\|\rho)
    \end{align}
\end{theorem}

To simplify the proof, we provide the following Pinsker-like bound on $D_R$. 
\begin{proposition}[$D_R$ and trace distance relationship]
    \label{prop:DR_Pinskers}
    For $\sigma, \rho \in \mc D_d$, 
    \begin{equation}
        D_R(\sigma\|\rho) \geq \frac{1}{4}\|\sigma -\rho\|_1^2
    \end{equation}
\end{proposition}
\begin{proof}
    \begin{align*}
        D_R(\sigma\|\rho) &\geq -\log F(\sigma, \rho) \geq 1 - F(\sigma, \rho) \geq \frac{1}{4} \|\sigma - \rho\|_1^2.
    \end{align*}
    The first inequality uses the fidelity bound in Remark \ref{remark:PropsOfDR}, and the last inequality is the Fuchs--van de Graaf inequality, e.g. \cite{Wilde_2016}. 
\end{proof}

\begin{proof}[Proof of Theorem \ref{thm:CompositeQuantumHypTestBound}]
    Suppose we have a quantum method of types, described by $\hat \Sigma_n, \mc M_n$, as in Theorem \ref{theorem:QuantMethodOfTypes}, with associated state estimation scheme $E_n$.
    
    Define 
    \begin{align*}
        A_n &:= \left\{ \hat \tau_n \in \mc D_d: \|\hat \tau_n - \sigma\|_1 \leq n^{-1/3}, \text{ for some }\sigma \in A\right\} \\
        T_n &:= A_n \cap \hat \Sigma_n\\
        T_n^C &:= A_n' \cap \hat \Sigma_n,
    \end{align*}
    where $A_n'$ denotes the complement of $A_n$. Let $L_n = E_n(T_n)$ be the effect for our test.

    Our type I error is given by 
    \begin{align*}
        \alpha_n(L_n) &= \sup_{\sigma \in A} \Tr[(1-L_n) \sigma^{\otimes n}]\\
        &\leq \sup_{\sigma \in A} \poly(n) \exp\left\{-n \min_{\tau \in T_n^C} D_R(\tau\|\sigma)\right\}\\
        &\leq \sup_{\sigma \in A} \poly(n) \exp\left\{-\frac{n}{4} \min_{\tau \in T_n^C} \|\tau - \sigma\|_1^2 \right\}\\
        &\leq \poly(n) \exp\left\{-\frac{n^{1/3}}{4}\right\} \to 0.
    \end{align*}
    where we made use of Proposition \ref{prop:DR_Pinskers} in the second inequality. 
    
    When analyzing the type II error, we get 
    \begin{align*}
        \beta_n(L_n) &= \sup_{\rho \in B} \mathbb{P}_{\mc M_n}[T_n|\rho]\\
        &\leq \poly(n) \exp\{-n \inf_{\substack{ \sigma \in A_n\\ \rho \in B}} D_R(\sigma\|\rho)\}  .
    \end{align*}
    Therefore, 
    \begin{equation*}
        \liminf_{n\to \infty} - \frac{1}{n} \log \beta_n(L_n) \geq \liminf_{n \to \infty}\  \inf_{\substack{\sigma \in A_n \\ \rho \in B}} D_R(\sigma\|\rho).
    \end{equation*}

    For any $n_1 \geq n_2$, we have $A_{n_1} \subset A_{n_2}$, and $\bigcap_{n=1}^{\infty} A_n = A$. Therefore, since $A_n \times B$ is compact, and by Remark \ref{remark:PropsOfDR} we know that $D_R$ is lower-semicontinuous, we get 
    \[
        \liminf_{n \to \infty} \inf_{\substack{\sigma \in A_n\\ \rho \in B}} D_R(\sigma\|\rho) = \inf_{\substack{\sigma \in A\\ \rho \in B}}D_R(\sigma\|\rho)
    \]
\end{proof}

We note that the result of Theorem \ref{thm:CompositeQuantumHypTestBound} should be interpreted as a simplification of what could already be obtained from existing large-deviation results \cite{Keyl_2006}. Our contribution is a finite-outcome measurement  obtaining the result. 

For comparison, quantum Stein's lemma \cite{hiai1991proper,ogawa2000strong} shows that the quantum relative entropy is the proper type II error rate for a simple null and alternative.
Versions of quantum Sanov's theorem \cite{bjelakovic2005quantum, Notzel2014Invariant} extend this to the composite null and simple alternative, showing that the infimum over the quantum relative entropy is still the proper error exponent. However, $D_R(\sigma\|\rho) < D(\sigma\|\rho)$ whenever $D(\sigma\|\rho) < \infty$ and $[\sigma, \rho] \neq 0$ \cite{GMPW_2026}. Recent work by Hayashi and Fang \cite{hayashi2026operational} shows that $D_R$ is the best achievable type II rate for certain composite hypothesis testing problems. In light of this, Theorem \ref{thm:CompositeQuantumHypTestBound} gives an alternative method of showing the achievability result in their work. 

\bibliographystyle{marcotomPB}
\bibliography{optimal}
\end{document}